\documentclass{article}
\usepackage[T1]{fontenc}
\usepackage{graphicx}
\usepackage{fullpage}
\usepackage{xcolor}
\usepackage{amsmath,amssymb,amsfonts,amsthm,enumerate}
\usepackage{nicefrac}
\usepackage{hyperref,cleveref}
\usepackage{authblk,orcidlink}

\newcommand*{\email}[1]{\href{mailto:#1}{\nolinkurl{#1}} }

\crefname{conjecture}{conjecture}{conjectures}
\Crefname{conjecture}{Conjecture}{Conjectures}
\usepackage[caption=false]{subfig}

\DeclareMathOperator*{\argmax}{arg\,max}

\newtheorem{theorem}{Theorem}[section]
\newtheorem{lemma}[theorem]{Lemma}

\newtheorem{conjecture}[theorem]{Conjecture}

\begin{document}

\title{Separating Coverage and Submodular: \\ Maximization Subject to a Cardinality Constraint}
\author{Yuval Filmus\,\orcidlink{0000-0002-1739-0872}\,}
\affil{The Henry \& Marilyn Taub Faculty of Computer Science and Faculty of Mathematics \\ Technion --- Israel Institute of Technology \\ \email{yuvalfi@cs.technion.ac.il}}
\author{Roy Schwartz}
\affil{The Henry \& Marilyn Taub Faculty of Computer Science\\ Technion --- Israel Institute of Technology \\ \email{schwartz@cs.technion.ac.il}}
\author{Alexander V. Smal\,\orcidlink{0000-0002-8241-5503}\,}
\affil{JetBrains Research \\ \email{avsmal@gmail.com}}

\maketitle

\begin{abstract}
We consider two classic problems: maximum coverage and monotone submodular maximization subject to a cardinality constraint.
[Nemhauser--Wolsey--Fisher '78] proved that the greedy algorithm provides an approximation of $1-\nicefrac{1}{e}$ for both problems, and it is known that this guarantee is tight
([Nemhauser--Wolsey '78; Feige '98]).
Thus, one would naturally assume that everything is resolved when considering the approximation guarantees of these two problems, as both exhibit the same tight approximation and hardness.

In this work we show that this is not the case, and study both problems when the cardinality constraint is a constant fraction $c \in (0,1]$ of the ground set.
We prove that monotone submodular maximization subject to a cardinality constraint admits an approximation of $1-(1-c)^{\nicefrac{1}{c}}$; This approximation equals $1$ when $c=1$ and it gracefully degrades to $1-\nicefrac{1}{e}$ when $c$ approaches $0$.
Moreover, for every $c=\nicefrac{1}{s}$ (for any integer $s\in \mathbb{N}$) we present a matching hardness.

Surprisingly, for $c=\nicefrac{1}{2}$ we prove that Maximum Coverage admits an approximation of $0.7533$, thus separating the two problems.
To the best of our knowledge, this is the first known example of a well-studied maximization problem for which coverage and monotone submodular objectives exhibit a different best possible approximation.
\end{abstract}

\section{Introduction}\label{sec:intro}
In this work we consider two related classic maximization problems. 
First, in Maximum Coverage (MC) we are given a (weighted) collection of elements $E$ equipped with $w\colon E\rightarrow \mathbb{R}_+$, sets $ S_1,\ldots,S_n\subseteq E$, and a cardinality bound $k\in \mathbb{N}$.
The goal is to find a collection of at most $k$ sets $ X\subseteq [n]$ that maximizes the total weight of covered elements: $ \sum _{e\in \bigcup _{i\in X}S_i }w_e$.

MC is captured by the well-studied problem of maximizing a nonnegative monotone submodular function subject to a cardinality constraint. Given a universe $U$ of size $n$, a set function $f\colon 2^U\rightarrow \mathbb{R}_+$ is \emph{submodular} if it satisfies $ f(A)+f(B)\geq f(A\cup B)+f(A\cap B)$ $ \forall A,B\subseteq U$,\footnote{An equivalent definition: $ f(A\cup \{ e\})-f(A) \geq f(B\cup \{ e\}) - f(B)$, $ \forall A\subseteq B$ and $ e\notin B$.} and is monotone if $ f(A)\leq f(B)$ $ \forall A\subseteq B$.
In this problem, denoted by SM, we are given $U$, a monotone submodular function $f$ represented by a value oracle,\footnote{A value oracle is a black box that allows the algorithm to access $f$ by querying the oracle with some $S\subseteq U$ and receiving from the oracle $f(S)$.} and a cardinality bound $ k\in \mathbb{N}$.
The goal is to find $ S\subseteq U$ of size at most $k$ that maximizes $f(S)$.
Clearly, SM captures MC since coverage functions are submodular and monotone.

SM was studied as early as the late $70$'s, when in their celebrated work, Nemhauser, Wolsey, and Fisher \cite{NWF78} proved that the simple greedy algorithm provides an approximation of $ 1-\nicefrac{1}{e}$.
Moreover, Nemhauser and Wolsey \cite{NW78} proved that this guarantee is tight, i.e., any algorithm that achieves an approximation better than $ 1-\nicefrac{1}{e}$ is required to perform exponentially many value queries.
Later, Feige \cite{F98} proved that assuming $ P\neq \mathit{NP}$ the bound of $ 1-\nicefrac{1}{e}$ is also tight for the special case of MC.

Both MC and SM have numerous theoretical and practical applications, e.g., the spread of influence in networks \cite{KKT05,KKT15}, social graph analysis \cite{NTMZMS18}, dictionary selection \cite{DK08}, data summarization \cite{MKZK18}, and many more.
Thus, it is no surprise that these problems have been studied extensively in different settings, e.g., fast and implementable algorithms \cite{BV14,BFS17,MBKVK15}, non-monotone submodular objectives \cite{BF19,BF24,BFNS14,EN16}, randomized vs.\ deterministic algorithms \cite{BF18,BF24focs,BF25}, and online algorithms \cite{CHJKT18,BFS19}, to name a few.

One would naturally assume that with respect to the approximation of both these classic and well-studied problems everything is resolved, for two main reasons.
First, matching upper and lower bounds are known.
Second, there is no known difference (to the best of our knowledge) between the best possible approximation when considering maximization, with any constraint, of coverage and monotone submodular objectives.
Another classic example of the latter is a knapsack constraint \cite{KMN99,S04}.

In this work we prove that this is not the case.
We study the behavior of the best possible approximation for both MC and SM as a function of the cardinality bound $k$.
Specifically, we assume that $ k=cn$ for some absolute constant $0<c\leq 1$, and aim to prove both upper and lower bounds on the possible approximation as a function of $c$.
To the best of our knowledge, this behavior has never been previously studied until this work.
The aforementioned works imply that regardless of the value of $c$, the approximation is always at least $ 1-\nicefrac{1}{e}$ for both MC and SM.
Moreover, for both problems a uniform random solution provides an approximation of at least $c$ (see, e.g., \cite[Lemma 2.2]{FMV11}).

A surprising result of our work is that MC and SM admit different approximations when studied in this regime.
For example, for $ c=\nicefrac{1}{2}$ we prove that MC admits an approximation that is strictly better than the hardness of SM.
This separates maximization of a coverage function when compared to maximization of a general monotone submodular function.
We are not aware of any previously known natural example of this unexpected phenomenon in the context of combinatorial optimization.
The situation can be compared to the difference between the optimal approximation ratios of MAX-CUT and (symmetric or arbitrary) non-monotone submodular maximization.

\subsection{Our Results}\label{sec:results}

\paragraph{Algorithms for MC and SM.}
First, for every constant $ 0<c\leq 1$ we present a simple LP-based approximation algorithm for MC whose approximation ratio depends on $c$. 
\begin{theorem}\label{thrm:LP-alg}
For every constant $0<c\leq 1$ there exists a polynomial time algorithm for Maximum Coverage with $ k=cn$ that achieves an approximation of $\rho(c)$.
If $ c=\nicefrac{1}{s}$ for some integer $ s\in \mathbb{N}$, then $ \rho(c)=1-(1-c)^{\nicefrac{1}{c}}$.
Otherwise, $ \nicefrac{1}{(s+1)}<c< \nicefrac{1}{s}$ for some integer $ s\in \mathbb{N}$ and $\rho(c)=1-\sigma(\alpha^*,s)$, where:
\begin{enumerate}[(1)]
\item $ \sigma(\alpha,m)=(1-\alpha c-(1-\alpha)/m)^m$.
\item $ \alpha^* \in (0,1)$ is the unique solution of the equation $ \sigma(\alpha^*,s)=\sigma(\alpha^*,s+1)$.
\end{enumerate}
\end{theorem}

\begin{figure}
    \centering
    \subfloat
    [Entire curve]
    {
         \includegraphics[width=0.45\textwidth]{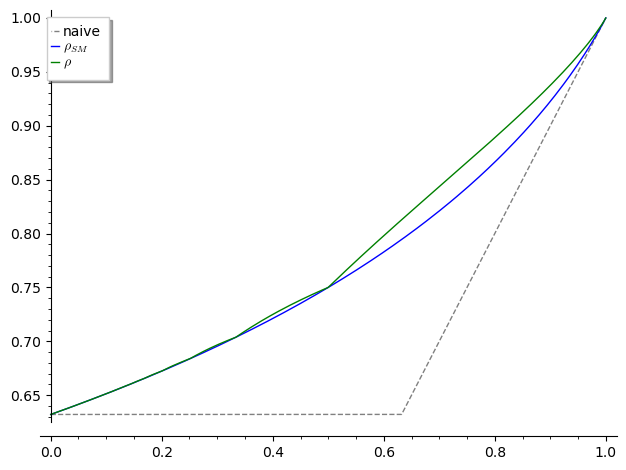}    
    }
    \hfill
    \subfloat
    [Zoom-in on {$c \in [1/2, 1]$}]
    {
         \includegraphics[width=0.45\textwidth]{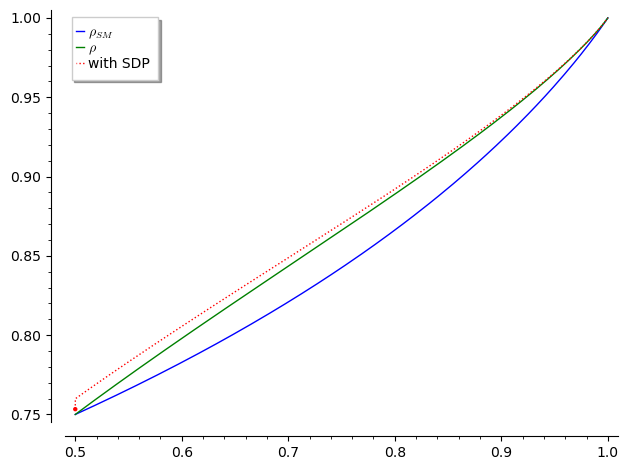}    
    }
    \caption{
        $1-(1-c)^{\nicefrac{1}{c}}$ (blue) is approximation for SM, $\rho(c)$ (green) is hardness for SM and approximation for MC, $\max\{1-\nicefrac1e, c\}$ (dashed gray) is naive algorithm, and red point (proven) extended by dotted red line (numerical under a generalization of Austrin's simplicity conjecture~\cite{A07}, see \Cref{app:SDP}) is approximation for MC. Both red dot and dotted red line separate MC and SM.
    }
    \label{fig:curves}
\end{figure}

\Cref{fig:curves} depicts $ \rho(c)$.
Some notable remarks regarding $ \rho(c)$:
\begin{enumerate}[(1)]
\item $ \rho(c)$ is an increasing function in $c$.
\item If $ c=\nicefrac{1}{s}$ for some integer $s\in \mathbb{N}$ then $ \rho(c)=1-(1-c)^{\nicefrac{1}{c}}$.
\item  $\rho(1)=1$, i.e., if $c=1$ our algorithm gives an optimal solution.
\item $ \lim _{c\rightarrow 0}\rho(c)=(1-\nicefrac{1}{e})$.
\item $ \rho(c) > \max \{ c, 1-\nicefrac{1}{e}\}$ for all $0 < c < 1$.
\end{enumerate}
An immediate consequence of the above is that our approximation is strictly better than the naive algorithm that takes the best of the two previously known algorithms. It is worth noting that for every $c$ we present a matching integrality gap of $\rho(c)$ for the LP that is used by our algorithm, thus proving the tightness of our analysis.

The algorithm underlying \Cref{thrm:LP-alg} relies on the explicit description of the MC instance. It can be modified to work even given only value oracle access to the instance.

\begin{theorem}\label{thrm:LP-alg-oracle}
For every constant $0<c\leq 1$ there exists a polynomial time algorithm for Maximum Coverage with $ k=cn$ that achieves an approximation of $\rho(c)$, where the input is provided by a value oracle.
\end{theorem}

Second, for every constant $ 0<c\leq 1$ we present a simple approximation algorithm for SM whose approximation ratio depends on $c$. The algorithm is implicit in~\cite{FNS11}, and we include it here for completeness.
\begin{theorem}\label{thrm:contgreedy-alg}
For every constant $ 0<c\leq 1$ there exists a polynomial time algorithm for monotone submodular maximization subject to a cardinality constraint $ k=cn$ that achieves an approximation of $ 1-(1-c)^{\nicefrac{1}{c}}$.
\end{theorem}

\Cref{fig:curves} depicts the curve $ 1-(1-c)^{\nicefrac{1}{c}}$.
Some notable remarks regarding this approximation guarantee:
\begin{enumerate}[(1)]
\item For every $ 0<c< 1$, $ 1-(1-c)^{\nicefrac{1}{c}}> \max \{ c,1-\nicefrac{1}{e}\} $, i.e., our algorithm is strictly better than the naive algorithm that takes the best of the two previously known algorithms.
\item If $ c=\nicefrac{1}{s}$ for some integer $s\in \mathbb{N}$ then $1-(1-c)^{\nicefrac{1}{c}}=\rho (c) $, thus \Cref{thrm:LP-alg,thrm:contgreedy-alg} provide the same guarantee for these values of $c$.
\item $ \rho(c)> 1-(1-c)^{\nicefrac{1}{c}}$ for all $c$ such that there does not exist an integer $s\in \mathbb{N}$ for which $ c=\nicefrac{1}{s}$, i.e., for these values of $c$, \Cref{thrm:LP-alg} provides an approximation strictly better than that of \Cref{thrm:contgreedy-alg}.
\end{enumerate}
One should recall that \Cref{thrm:contgreedy-alg} is for SM, whereas \Cref{thrm:LP-alg} is for MC, which is a special case of SM.

\paragraph{Value Oracle Hardness for SM.}
Third, we present a hardness result for SM that equals $\rho(c)$ and (almost) matches the guarantee of \Cref{thrm:contgreedy-alg} for every $c$. The case $c=\nicefrac1s$ of the theorem essentially appears already in~\cite{V13}.

\begin{theorem}\label{thrm:SM-hardness}
For every rational constant $ 0<c\leq 1$ and every $ \varepsilon >0$, any (possibly randomized) algorithm for monotone submodular maximization subject to a cardinality constraint $ k=cn$ that achieves an approximation of $(1+\varepsilon)\rho(c)$ requires exponentially many value queries.    
\end{theorem}

It is worth nothing that \Cref{thrm:SM-hardness,thrm:contgreedy-alg} resolve the approximability of SM when $ c=\nicefrac{1}{s}$ for $ s\in \mathbb{N}$ as both the algorithm and the hardness match (recall that in this case $ \rho(c)=1-(1-c)^{\nicefrac{1}{c}}$). 

We do not prove similar hardness results for MC, as we actually prove that MC admits a better approximation.
Moreover, the results of Austrin and Stankovi\'c~\cite{AS19} imply that MC is APX-hard for all $c \in (0,1)$, even in the special case of Max $k$-Vertex-Cover, in which each element is covered by exactly two sets.

\paragraph{Separating Coverage and Submodular Maximization.}
Perhaps the most surprising result is that we are able to \emph{separate} MC and SM, even when in both cases the instance is given via a value oracle, as the following theorem summarizes.

\begin{theorem}\label{thrm:separating}
There exists a polynomial time algorithm for Maximum Coverage with $ k=\nicefrac{n}{2}$ that achieves an approximation of at least $0.7533-o(1)$.
The algorithm only requires value oracle access to the instance.
\end{theorem}

It is important to note that the above theorem implies that MC, when $c=\nicefrac{1}{2}$, has an approximation that is strictly better than $\rho(\nicefrac{1}{2})=\nicefrac{3}{4}$.
Recall that the latter is the hardness for SM, when $ c=\nicefrac{1}{2}$, as \Cref{thrm:SM-hardness} states.
Refer to \Cref{fig:curves} for this result.

Moreover, assuming a generalization of the simplicity conjecture of Austrin~\cite{A07}, this separation extends (numerically) to all constants $ \nicefrac{1}{2}< c< 1$, showing that MC with $ k=cn$ admits an approximation that is \emph{strictly} better than $\rho(c)$ by an additive absolute constant.

\subsection{Our Techniques}\label{sec:techniques}
A main ingredient in both algorithms for MC, i.e., \Cref{thrm:LP-alg,thrm:separating}, is that we round fractional solutions that are \emph{not} optimal with respect to our relaxations.
We average a fractional optimal solution, LP in the case of \Cref{thrm:LP-alg} and both LP and SDP in the case of \Cref{thrm:separating}, with a fractional solution that corresponds to a uniform random solution.
The mixing probability needs to be carefully chosen, as it depends on both $c$ and the rounding algorithm we use for the LP or SDP. The SDP is used to handle the elements in the instance which are covered exactly twice. These elements form an instance of Max $k$-Vertex-Cover, a problem for which Austrin, Benabbas and Georgiou~\cite{ABG16} gave an SDP-based approximation algorithm.

Moreover, in order to obtain (numerically) better algorithms for MC when $ \nicefrac{1}{2}<c<1$ (conditional on a generalization of Austrin's simplicity conjecture), our algorithm generalizes Austrin's algorithm for the case $c = \nicefrac12$ by handling arbitrary $c$ and by allowing elements to be covered only once.

Our results for SM are based on the measured continuous greedy algorithm \cite{CCPV11,FNS11}.
This algorithm simulates a continuous process.
The longer the process runs, the higher the value of the output is.
However, if the process runs for too long, we might end up with an infeasible solution.
We exploit the simple structure of the cardinality constraint to show that there exists a good stopping time for this process.

Our hardness results for SM are based on the symmetry gap machinery of Vondr{\'{a}}k \cite{V13}, in which a fixed instance with a gap between symmetric and arbitrary solutions is lifted to value oracle hardness.
This lifting process maintains the ratio $k/n$.
Curiously, the instances used to prove hardness for SM are identical to those that are used to show a tight integrality gap for the LP of \Cref{thrm:LP-alg}.

\subsection{Related Work}\label{sec:related-work}

While we are unaware of previous work that studies either MC or SM in the regime $k = cn$, Buchbinder et al.~\cite{BFNS14} study \emph{non-monotone} submodular maximization subject to a cardinality constraint $k \leq cn$ or $k = cn$ (in contrast to the monotone setting, in the non-monotone setting the two constraints are different). In both cases, their algorithm gives an approximation ratio of $(1 + \nicefrac{1}{2\sqrt{c(1-c)}})^{-1} - o(1)$ for all $c \leq \nicefrac12$, and $\nicefrac12 - o(1)$ for all $c > \nicefrac12$.

Another related work in the same vein due to Li et al.~\cite{LFKK22} studies the problem of maximizing a monotone linear function subject to the constraint $k \leq cn$, when the linear function is given as a value oracle. The paper, which focuses on algorithms making very few queries, shows that any algorithm whose approximation ratio beats $c$ must make $\Omega(\nicefrac{n}{\log n})$ value oracle queries.

Separations between maximum coverage and monotone submodular functions do appear in other settings. One example is mechanism design: Dughmi, Roughgarden and Yan~\cite{DRY11} give a $1-1/e$-approximation mechanism for welfare maximization when the valuations of the bidders are matroid rank sums, a class of functions thatincludes coverage functions. In contrast, Dughmi and Vondr\'ak~\cite{DV11} and Dobzinski and Vondr\'ak~\cite{DV12} rule out any constant approximation for general monotone submodular valuations.

\subsection{Paper Organization}\label{sec:organization}
\Cref{sec:algorithms} contains the algorithms for proving \Cref{thrm:LP-alg,thrm:LP-alg-oracle,thrm:contgreedy-alg}.
The hardness for SM, \Cref{thrm:SM-hardness}, appears in \Cref{sec:hardness}, where we also give a tight integrality gap for the LP underlying \Cref{thrm:LP-alg}.
The separation of MC and SM for $c = \nicefrac12$ is proved in \Cref{sec:separating}. The extension to $c \in (\nicefrac12, 1)$ appears in \Cref{app:SDP}.

\paragraph{Acknowledgements} We thank StackExchange user John for proposing the problem~\cite{54170}, and Zach Chase for useful discussions. We thank the anonymous IPCO reviewers for many helpful suggestions, including one that led to \Cref{thrm:LP-alg-oracle}.
Yuval Filmus has received funding from the European Union’s Horizon 2020 research and innovation program under grant agreement no.\ 802020-ERC-HARMONIC, and by ISF grant 507/24. Roy Schwartz and Alexander V.\ Smal have received funding from the European Union’s Horizon 2020 research and innovation program under grant agreement no.\ 852870-ERC-SUBMODULAR.

\section{Preliminaries}\label{sec:prelim}
Given a set function $f\colon 2^U\rightarrow \mathbb{R}_+$, the \emph{multilinear} extension $ F\colon [0,1]^U\rightarrow \mathbb{R}_+$ is defined as follows: $ F(x)\triangleq \sum _{T\subseteq E}f(T) \prod _{e\in T}x_e \prod _{e\notin T}(1-x_e)$, $ \forall x\in [0,1]^U$.
Intuitively, $F(x)$ has a probabilistic meaning: it is the expected value of $f$ on a random set which includes each element $e
\in E$ independently with probability $x_e$.

We require an algorithm for rounding a fractional point with respect to the multilinear extension. 
\begin{theorem}[Pipage Rounding \cite{AS04,CCPV07}]\label{thrm:pipage}
There exists a polynomial time algorithm that, given a matroid $ M=(U,\mathcal{I})$ and a point $ \mathbf{z}\in \mathcal{P}$ (where $ \mathcal{P}$ is the independent sets polytope associated with $M$), finds $ \mathbf{r}\in \mathcal{P}$ satisfying:
\begin{enumerate}[(1)]
\item $\mathbf{r}\in \mathcal{P} \cap \{ 0,1\} ^U$.
\item $ F(\mathbf{r}) \geq F(\mathbf{z})$.
\end{enumerate}
\end{theorem}
We note that instead of pipage rounding one can use, e.g., a deterministic version of swap rounding~\cite{CVZ10}.

\smallskip

The Max $k$-Vertex-Cover problem is a special case of MC.
Given an undirected graph $G$ on $n$ vertices, the elements $E$ are the edges of the graph, and each vertex $v$ corresponds to a set that contains all edges touching $v$.
Thus, the goal (as in MC) is to choose $X\subseteq V$ of size at most $k$ that maximize the total weight of edges covered by $X$, i.e., edges that have at least one endpoint in $X$.
We utilize the following result of Austrin, Benabbas, and Georgiou \cite{ABG16} for the case that $ c=\nicefrac{1}{2}$ (this result applies to the more general problem of Max $2$-CNF where exactly half the variables are required to be true).
The result of \cite{ABG16} is based on semi-definite programming and on the works of Lewin, Livant, and Zwick \cite{LLZ02} and Raghavendra and Tan \cite{RT12}. It is based on a certain \emph{simplicity conjecture}, recently proved by Brakensiek, Huang and Zwick~\cite{BHZ24}.
\begin{theorem}[Max $k$-Vertex-Cover \cite{ABG16}]\label{thrm:MaxkVC}
There is a polynomial time algorithm for Max $k$-Vertex-Cover with $ k=\nicefrac{n}{2}$ that achieves an approximation of at least $ 0.9401$.
\end{theorem}

\section{Algorithms for MC and SM}\label{sec:algorithms}

\subsection{Algorithm for MC}\label{sec:algMC}
Our algorithm is based on a simple ``non-oblivious'' rounding of the natural LP, i.e., it rounds an altered solution and not the optimal solution for the LP.
We start by stating the classic LP for MC (see, e.g., \cite{AS99}):
\[ \max \left\{ \sum _{e\in E}w_e x_e: x_e\leq \min \left\{ 1,\sum _{i\colon e\in S_i}y_i\right\} ~\forall e\in E, ~ \sum _{i=1}^n y_i\leq cn, ~ y_i\geq 0 ~ \forall i\in [n] \right\}. \]
In the above, every $ y_i$ indicates whether $S_i$ is chosen, and $x_e$ indicates whether $e$ is covered.
Clearly, for an optimal solution to this LP we have that $x_e\geq 0$, $\forall e\in E$, and $ y_i\leq 1$, $ \forall i\in [n]$.
From this point onwards we assume that $ \{x_e\} _{e\in E}$ and $\{ y_i\}_{i\in [n]}$ are an optimal solution to the above LP.

We define a new set of variables $z_i$, $ \forall i\in [n]$, by averaging the LP's solution with a uniform random fractional solution:
\[ z_i \triangleq \alpha c+(1-\alpha)y_i, \]
where $\alpha$ depends on $c$ and will be chosen later.
Our main observation is summarized in the following lemma.
It states that if every set $S_i$ is chosen independently with a probability of $z_i$ then the probability that element $e\in E$ is covered is at least $ \rho(c) x_e$.
We state the lemma via the multilinear extension $F$ of the set function $f\colon 2^{[n]}\rightarrow \mathbb{R}_+$ which is the weighted coverage function of the given instance of MC, i.e., $ f(T)\triangleq \sum _{e\in E}w_e\mathbf{1}_{\{ e\in \cup _{i\in T}S_i\}}$, $ \forall T\subseteq [n]$.

\begin{lemma}\label{lem:AlgMultilinear}
$ F(\mathbf{z})\geq \rho (c) \sum _{e\in E}x_e$, where $\rho(c)$ is defined as in \Cref{thrm:LP-alg}.    
\end{lemma}
\begin{proof}
Recalling the probabilistic interpretation of $F$ (see \Cref{sec:prelim}), it suffices to fix an element $ e\in E$ and prove that: $ 1-\prod _{i:e\in S_i}(1-z_i)\geq \rho(c) x_e$.
For simplicity of presentation assume that $e$ belongs to $ S_1,\ldots,S_m$.
Hence,
\begin{align}
1-\prod _{i=1}^m (1-z_i) & =   1-\prod _{i=1}^m (1-\alpha c - (1-\alpha )y_i)  \label{ineq1}\\
& \geq 1 - (1-\alpha c - (1-\alpha)(y_1+\dots+y_m)/m)^m \label{ineq2}\\
& \geq 1 - (1-\alpha c - (1-\alpha)x_e/m)^m \label{ineq3}\\
& \geq x_e ( 1 - (1-\alpha c - (1-\alpha)/m)^m) . \label{ineq4}
\end{align}
In the above, \eqref{ineq1} is by the definition of $z_i$, \eqref{ineq2} is just the AM-GM inequality, \eqref{ineq3} follows from the constraint upper bounding $x_e$ in the LP, and \eqref{ineq4} follows from the concavity of $  1 - (1-\alpha c - (1-\alpha)x_e/m)^m $ in $x_e$ when $ m\geq 1$ (using $x_e \le 1$).

Given $c$, all that is left to conclude the proof is to choose $\alpha$ and prove that:
\[ \min _{m\geq 1}\{ 1 - (1-\alpha c - (1-\alpha)/m)^m\}\geq \rho(c). \]
We denote $ \sigma (\alpha,m)\triangleq (1-\alpha c - (1-\alpha)/m)^m = (1 - \nicefrac1m - (c - \nicefrac1m)\alpha)^m$ and $ \sigma (\alpha)\triangleq \max _{m\geq 1}\{ \sigma(\alpha,m)\}$, where the maximum is taken over integer $m$.
Note that $\sigma(\alpha,m)$ is decreasing in $\alpha$ when $ m>\nicefrac{1}{c}$, increasing in $\alpha$ when $ m<\nicefrac{1}{c}$, and equals $(1-\nicefrac{1}{m})^m$ when $ m=\nicefrac{1}{c}$.
Moreover, we note that $\sigma(\alpha,m)$ is log-concave in $m$ since
$
 \frac{\partial^2 \log \sigma(\alpha, m)}{\partial m^2} = -\frac{(1-\alpha)^2/m}{[(1-\alpha c)m - (1-\alpha)]^2} < 0
$.

Our proof considers two cases, according to whether $ c=\nicefrac{1}{s}$ for some integer $ s\in \mathbb{N}$ or not.
First, assume that $c=\nicefrac{1}{s} $ for some $ s\in \mathbb{N}$, where $s \ge 2$ (the case $c=1$ is trivial).
For any fixed $\alpha$, denote by $ m^*(\alpha)$ the (not necessarily integer) maximizer of $ \sigma (\alpha,m)$, i.e., $ m^*(\alpha)\triangleq \argmax _{m\geq 1} \{ \sigma (\alpha,m)\}$.
We prove that in this case, choosing $ \alpha ^* =1-(s-1)\ln{(\nicefrac s{(s-1)})}$ implies that $ m^*(\alpha^*)=s$; note $\alpha^* \in (0,1)$ since $\ln{(\nicefrac s{(s-1)})} < \nicefrac1{(s-1)}$.
Thus, $ \sigma(\alpha^*,m^*(\alpha))=(1-c)^{\nicefrac{1}{c}}$, since $ m^*(\alpha^*)=s$, $s=\nicefrac{1}{c}$, and $ \sigma(\alpha^*,\nicefrac{1}{c})=(1-c)^{\nicefrac{1}{c}}$.

Recall that $\sigma(\alpha,m)$ is log-concave in $m$, thus to find $ m^*(\alpha^*)$ it suffices to find the (unique) $m$ for which $ \partial(\log \sigma(\alpha^*,m))/\partial m=0$.
Direct calculation of the latter shows that it is equivalent to:
\[ \ln{(1-\alpha^* c - (1-\alpha^*)/m)} +
\frac{(1-\alpha^*)/m}{1-\alpha^* c - (1-\alpha^*)/m} = 0. \]
Plugging $ c=\nicefrac{1}{s}$, $ m=s$, and $ \alpha^*=1-(s-1)\ln{(\nicefrac s{(s-1)})}$, we have $1 - \alpha^*c - {(1-\alpha^*)/m} = \nicefrac{(s-1)}s$, from which the equality is easy to verify.

Second, we assume that $c$ satisfies $ \nicefrac{1}{(s+1)}<c<\nicefrac{1}{s}$ for some integer $ s\in \mathbb{N}$.
Define the following function $ \delta (\alpha)\triangleq \sigma(\alpha,s+1)-\sigma (\alpha,s)$.
Note that $ \delta (\alpha)$ is a decreasing function in $\alpha$ since: 
\begin{enumerate}[(1)]
\item $\sigma(\alpha,m)$ is increasing when $ m<\nicefrac{1}{c}$ (and $ s<\nicefrac{1}{c}$).
\item $ \sigma(\alpha,m)$ is decreasing when $ m>\nicefrac{1}{c}$ (and $ s+1>\nicefrac{1}{c}$).
\end{enumerate}
Moreover:
\[ \delta (0)=(1-\nicefrac{1}{(s+1)})^{s+1}-(1-\nicefrac{1}{s})^s > 0 ~\text{ and }~ \delta(1)=(1-c)^{s+1}-(1-c)^s <0. \]
Therefore, there exists a (unique) $ \alpha^*$ for which $ \delta (\alpha^*)=0$.
This $ \alpha^*$ satisfies: $ \sigma(\alpha^*,s+1) = \sigma(\alpha^*,s)$.
Moreover, since $ \sigma(\alpha^*,m)$ is log-concave in $m$, it must be the case that $ s<m^*(\alpha^*)<s+1$.
Thus, we proved that $\sigma(\alpha^*)=\sigma(\alpha^*,s)=\sigma(\alpha^*,s+1) < \sigma(\alpha^*,m^*(\alpha^*))$.
This concludes the proof of the lemma.
\end{proof}

Note that in the second case, $\rho(c) > 1-\sigma(\alpha^*, \nicefrac1c) = 1-(1-c)^{\nicefrac1c}$. We give a matching integrality gap for the LP in \Cref{sec:integrality-gap}.

\begin{proof}[Proof of \Cref{thrm:LP-alg}]
Our algorithm performs the following steps: 
\begin{enumerate}[(1)]
\item Find an optimal solution $ \{ x_e\}_{e\in E}$ and $ \{ y_i\}_{i\in [n]}$ to the LP. 
\item Compute $ z_i\triangleq \alpha c+ (1-\alpha)y_i$ $\forall i\in [n]$, where $\alpha$ is chosen as in the proof of \Cref{lem:AlgMultilinear}. 
\item Apply \Cref{thrm:pipage} on $ \mathbf{z}$ and $\mathcal{P}\triangleq  \{ \mathbf{x}\in [0,1]^{[n]}:\sum _{i=1}^n x_i \leq cn\}$ (note that $ \mathcal{P}$ is the independent set polytope of the uniform matroid defined on $ [n]$ with $ k=cn$).
\end{enumerate}
Step (3) can be applied since $ \mathbf{z}\in \mathcal{P}$.
Thus, \Cref{lem:AlgMultilinear}, alongside the observation that $ \sum _{e\in E}w_e x_e$ is an upper bound on the value of an optimal solution, suffice to conclude the proof.
\end{proof}
It is worth noting that the approximation ratio $\rho(c)$ has an elegant closed form solution when $ \nicefrac{1}{2}<c<1$:
\[ \rho(c) = 1-\frac{(1-c)\left( 1-2\sqrt{c(1-c)}\right)}{(2c-1)^2}, ~~~~~~~\forall c\in (\nicefrac{1}{2},1).\]

To see this, we first determine $\alpha^*$ by solving the equation $\sigma(\alpha^*,1) = \sigma(\alpha^*,2)$. This is a quadratic equation $(1-c)\alpha^* = (\nicefrac12-(c-\nicefrac12)\alpha^*)^2$ whose solutions are $\frac{1 \pm 2\sqrt{c(1-c)}}{(2c-1)^2}$.
Since $1 > (1-2c)^2$, the only solution in $[0,1]$ is $\alpha^* = \frac{1 - 2\sqrt{c(1-c)}}{(2c-1)^2}$. We then compute $\rho(c) = 1-\sigma(\alpha^*,1) = 1-(1-c)\alpha^*$.

\subsection{Algorithm for MC in the value oracle model}\label{sec:algMC-oracle}

The algorithm underlying \Cref{sec:algMC} seems to rely on knowing the set system explicitly. However, using simple considerations, we can make do given only value oracle access to $f$. The essential observation is that for a given set $I \subseteq [n]$, the total weight of elements belonging to all sets $S_i$ for $i \in I$ and to no set $S_j$ for $j \notin I$ is
\[
 w_I := \sum_{J \subseteq I} (-1)^{|I \setminus J|} f(J).
\]
This formula follows easily from M\"obius inversion.

In particular, for every $M$ we can determine, using $O(n^M)$ value oracle queries, the value of $w_I$ for all $|I| \leq M$, as well as
\[
 w' := \sum_{|I| > M} w_I = f([n]) - \sum_{|I| \leq M} w_I.
\]

We now solve the following LP:
\begin{align*}
 \max \bigg\{ \sum_{|I| \leq M} w_I x_I + w' x' :
 &x_I \leq \min \left\{ 1, \sum_{i \in I} y_i \right\} ~\forall |I| \leq M, \\ &x' \leq \min \left\{ 1, \sum_{i=1}^n y_i \right\},
 \sum_{i=1}^n y_i \leq cn, ~ y_i \geq 0 ~ \forall i \in [n] \bigg\}.
\end{align*}

This LP corresponds to the instance in which all elements belonging to more than $M$ sets are assumed to belong to all sets. In particular, if $O$ is an optimal solution to the original instance, then it is a feasible solution to the LP, and its objective value is at least $f(O)$.

It remains to prove an analog of \Cref{lem:AlgMultilinear}.
The proof of the lemma boils down to the following two inequalities:
\begin{gather*}
    1 - \prod_{i \in I} (1 - z_i) \geq \rho(c) x_I, \\
    \sum_{|I| > M} w_I \left(1 - \prod_{i \in I} (1 - z_i)\right) \geq \rho(c) w' x'.
\end{gather*}
The first inequality is proved in \Cref{lem:AlgMultilinear}. Since $x' \leq 1$, the second inequality will follow from
\[
 |I| > M \Longrightarrow 1 - \prod_{i \in I} (1 - z_i) \geq \rho(c).
\]
To this end, we observe that $\alpha > 0$ (this is proved as part of the proof of \Cref{lem:AlgMultilinear}), and so $1 - \prod_{i \in I} (1 - z_i) \geq 1-(1-\alpha c)^{M + 1} \geq \rho(c)$ for some constant $M = M(c)$, completing the proof of the generalized version of \Cref{lem:AlgMultilinear} and so of \Cref{thrm:LP-alg-oracle}.

One drawback of this algorithm is that its running time $O(n^M)$ depends on $c$, where $M$ blows up as $c \to 1$. In contrast, the algorithm behind \Cref{thrm:LP-alg} has the same running time for all values of $c$.

\subsection{Algorithm for SM}\label{sec:algSM}
Our algorithm for SM follows the same type of analysis used by Feldman, Naor, and Schwartz \cite{FNS11} when applying the measured continuous greedy to a partition matroid independence constraint.
The measured continuous greedy algorithm of Feldman, Naor, and Schwartz \cite{FNS11} is based on the celebrated continuous greedy algorithm of Calinescu, Chekuri, P\'{a}l, and Vondr\'{a}k \cite{CCPV11}.

Both algorithms (approximately) solve the following problem: \[ \max \{ F(\mathbf{x}):\mathbf{x}\in \mathcal{P}\}, \]
where $ \mathcal{P}\subseteq [0,1]^U$ is a down-monotone polytope.\footnote{$\mathcal{P}$ is down-monotone if $\mathbf{x}\in \mathcal{P}$ and $ 0\leq \mathbf{y}\leq \mathbf{x}$ imply $ \mathbf{y}\in \mathcal{P}$.} 
The measured continuous greedy algorithm can be viewed as a continuous time process whose location at time $t$ is $ \mathbf{y}(t)$.\footnote{This continuous time process can be implemented as an algorithm by choosing a small step size and discretizing the movement of the process. This incurs a small additive $o(1)$ loss in the approximation (we refer the reader to \cite{CCPV11,FNS11} for the details).}
It is defined as follows: $ \mathbf{y}(0)\leftarrow 0$ and
\begin{align} \frac{\partial y_i(t)}{\partial t} \leftarrow (1-y_i(t))x_i(t)~~~\forall i\in U,\forall t\geq 0,\label{contgreedy1}
\end{align}
where $ \mathbf{x}(t)\leftarrow \argmax \{ \mathbf{x}\cdot \nabla F(\mathbf{y}(t)):\mathbf{x}\in \mathcal{P}\}$.
It is proved in \cite{FNS11} that:
\begin{align}F(\mathbf{y}(t))\geq (1-e^{-t})f(S^*)~~~\forall t\geq 0,\label{contgreedy2}
\end{align}
where $S^*\subseteq U$ is an optimal solution to the problem.

In our case $\mathcal{P}$ is just the uniform matroid polytope, i.e., $ \mathcal{P}=\{ \mathbf{x}\in [0,1]^U:\sum _{i\in U}x_i\leq k\}$ (recall that $ k=cn$), and $S^*$ is an integral optimal solution, i.e., $ S^*=\argmax \{ f(S):S\subseteq U, |S|\leq k\}$.
Clearly, \eqref{contgreedy2} implies that one would like to terminate the process as late as possible since this would result in a solution with higher value.
However, \eqref{contgreedy1} implies that the longer the process runs the larger $ \mathbf{y}(t)$ becomes, and at some point it might move outside of $\mathcal{P}$.
Thus, our goal is to find the largest stopping time $ T_c$ for which we are guaranteed that $ \mathbf{y}(T_c)\in \mathcal{P}$.

We are now ready to prove \Cref{thrm:contgreedy-alg}.
\begin{proof}[Proof of \Cref{thrm:contgreedy-alg}]
We note that at time $ t=0$: $ \sum _{i\in U}y_i(0)=0$.
Moreover, one can verify that for every $t\geq 0$ and $ i\in U$: \[ y_i(t)\leq 1-e^{-\int_{s=0}^t x_i(s)ds}. \]
Since the function $ 1-e^{-t}$ is concave in $t$, $ \sum _{i\in U}y_i(t)$ is maximized when every $ y_i(s)$, for every $ 0\leq s\leq t$, is increased at the same rate of $\nicefrac{k}{n}=c$ (recall that $ \mathbf{x}(s)\in \mathcal{P}$, i.e., $ \sum _{i\in U}x_i(s)\leq k$).
Thus, we can conclude that $ \sum _{i\in U}y_i(t)\leq n(1-e^{-ct})$.
Choose $T_c$ to be the time for which this upper bound is tight, i.e., $ 1-e^{-cT_c}=\nicefrac kn=c$, implying that $ T_c=\ln(\nicefrac{1}{(1-c)})/c$.
Plugging this $T_c$ into \eqref{contgreedy2} yeilds that $ F(\mathbf{y}(T_c))\geq (1-(1-c)^{\nicefrac{1}{c}})f(S^*)$.
Applying pipage rounding (\Cref{thrm:pipage}) concludes the proof.
\end{proof}

\section{Value Oracle Hardness for SM}\label{sec:hardness}

We show that no polynomial time algorithm for SM gives an approximation ratio better than $\rho(c)$ using Vondr\'ak's symmetry gap technique~\cite{V13}. We only describe the special case of a cardinality constraint (see the first example in \cite[Section 2]{V13}).

\begin{theorem}[{Symmetry Gap \cite{V13}}] \label{thrm:symmetry-gap}
Let $f\colon 2^U \to \mathbb{R}_+$ be a monotone submodular function with multilinear extension $F\colon [0,1]^U \to \mathbb{R}_+$, and suppose that $f$ is symmetric under the action of some symmetry group $\mathcal{G} \subseteq \operatorname{Sym}(U)$. Let $\mathrm{OPT} = \max \{ f(S) : |S| = c |U| \}$, let $\overline{\mathrm{OPT}} = \max \{ F(x) : \sum_{e \in U} x_e = c|U|, \allowbreak x \text{ is invariant under } \mathcal{G}\}$, and let $\rho = \overline{\mathrm{OPT}}/\mathrm{OPT}$ be the \emph{symmetry gap}.

Then for every $\epsilon > 0$, any (possibly randomized) algorithm for monotone submodular maximization subject to a cardinality constraint $k = cn$ that achieves an approximation of $(1+\epsilon)\rho$ requires exponentially many value queries.
\end{theorem}

In our applications, $f$ will always be a coverage function.

\begin{proof}[Proof of \Cref{thrm:SM-hardness}]
As in the proof of \Cref{thrm:LP-alg}, we consider two cases: $c = \nicefrac1s$ for $s \in \mathbb{N}$, and $\nicefrac1{(s+1)} < c < \nicefrac1s$ for $s \in \mathbb{N}$.

The case $c = \nicefrac1s$ is the second example in \cite[Section 2]{V13}, but we repeat it here for completeness. Let $f\colon 2^{\{1,\dots,s\}} \to \mathbb{R}_+$ be the coverage function corresponding to $s$ sets covering a common single element of weight~$1$, with multilinear extension $F(x_1,\dots,x_s) = 1 - (1-x_1)\cdots(1-x_s)$. Then $\mathrm{OPT} = 1$ and $\overline{\mathrm{OPT}} = F(\nicefrac1s, \dots, \nicefrac1s) = 1 - (1-\nicefrac1s)^s$, and so $\rho = 1 - (1-\nicefrac1s)^s = \rho(c)$. \Cref{thrm:symmetry-gap} completes the proof.

Second, suppose that $c$ satisfies $\nicefrac1{(s+1)} < c < \nicefrac1s$. Let $c = N/D$, and define $a = (s+1-\nicefrac1c)N \in \mathbb{N}_+$ and $b = (\nicefrac1c-s)N \in \mathbb{N}_+$. We construct a coverage function $f\colon 2^{U} \to \mathbb{R}_+$: there are $a$ elements of type $A$ of weight $p/a$ each, where $p \in [0,1]$ is a parameter, and $b$ elements of type $B$ of weight $(1-p)/b$ each. Each element of type $A$ is covered by $s$ many singleton sets of type $A$, and each element of type $B$ is covered by $s+1$ many singleton sets of type $B$. The total number of sets is $|U| = sa + (s+1)b$. We note that $c|U| = N = a+b$ and so $\mathrm{OPT} \geq a \cdot (p/a) + b \cdot ((1-p)/b) = 1$ by covering each element exactly once. To complete the proof, we will choose $p$ so that $\overline{\mathrm{OPT}} = \rho(c)$.

The function $f$ is invariant under a group $\mathcal{G}$ whose two orbits are (1) all sets of type $A$ and (2) all sets of type $B$. Consequently a solution is invariant under $\mathcal{G}$ if it gives the same value $x$ to all $sa$ sets of type $A$ and the same value $y$ to all $s(b+1)$ sets of type $B$. So $\overline{\mathrm{OPT}} = \max \{ g(x,y) : 0 \leq x,y \leq 1, \tfrac{a}{a+b}sx + \tfrac{b}{a+b}(s+1)y = 1 \}$ where
$g(x,y) = p(1-(1-x)^s) + (1-p)(1-(1-y)^{s+1})$.

Recall the function $\sigma(\alpha,m) = (1-\alpha c-(1-\alpha)/m)^m = (1-1/m-\alpha(c-1/m))^m$ from \Cref{sec:algMC}. We can write $g(x,y) = p(1-\sigma(\alpha,s)) + (1-p)(1-\sigma(\beta,s+1))$ where $\alpha,\beta$ satisfy $x = \nicefrac1s - \alpha(\nicefrac1s-c)$ and $y = \nicefrac1{(s+1)} + \beta(c-\nicefrac1{(s+1)})$. Solving for $\alpha,\beta$ gives $\alpha = \frac{\nicefrac1s-x}{\nicefrac1s-c}$ and $\beta = \frac{y-\nicefrac1{(s+1)}}{c-\nicefrac1{(s+1)}}$.

We claim that $\alpha = \beta$. Indeed, recall that $\tfrac{a}{a+b}sx + \tfrac{b}{a+b}(s+1)y = 1$, so $\frac{as}{a+b}(x-\nicefrac1s) + \frac{b(s+1)}{a+b}(y-\nicefrac1{(s+1)}) = 0$, so $\frac{\nicefrac1s-x}{b(s+1)} = \frac{y-\nicefrac1{(s+1)}}{as}$. The claim follows from $b(s+1) = (s(s+1)/c)(\nicefrac1s-c)$ and $as = (s(s+1)/c)(c-\nicefrac1{(s+1)})$.

In \Cref{sec:algMC} we showed that there is a unique value $\alpha^* \in [0,1]$ such that $\sigma(\alpha^*,s) = \sigma(\alpha^*,s+1)$. The corresponding values $x^*,y^*$ (obtained by substituting $\alpha,\beta := \alpha^*$ in the expressions for $x,y$) satisfy $x^* \in [c,\nicefrac1s]$ and $y^* \in [\nicefrac1{(s+1)},c]$, and so this is a feasible solution for $\overline{\mathrm{OPT}}$. In the remainder of the proof, we will find a value of $p \in [0,1]$ for which $g(x,y)$ is maximized at $(x^*,y^*)$, and so $\overline{\mathrm{OPT}} = 1-\sigma(\alpha^*,s) = \rho(c)$. \Cref{thrm:symmetry-gap} will then complete the proof.

Let $G(x) = g(x,y(x))$, where $y(x)$ is the unique value such that $\frac{a}{a+b}sx + \frac{b}{a+b}(s+1)y(x) = 1$. Straightforward computation shows that $G'' \leq 0$ and so $G$ is concave. In order for $x^*$ to be the maximum, we need $G'(x^*) = 0$. Equivalently, 
$b(s+1) \cdot \partial g/\partial x = as \cdot \partial g/\partial y$, which leads to the equation:
\[
 b(s+1) \cdot p s(1-x^*)^{s-1} = as \cdot (1-p) (s+1) (1-y^*)^s.
\]
Since $(1-x^*)^s = \sigma(\alpha^*,s) = \sigma(\alpha^*,s+1) = (1-y^*)^{s+1}$, this is equivalent to $p\frac{b}{1-x^*} = (1-p)\frac{a}{1-y^*}$. When $p = 0$, the RHS is larger, and when $p = 1$, the LHS is larger. Hence there is a (unique) $p$ for which both sides are equal. For this value of $p$ we get $G'(x^*) = 0$ and so $(x^*,y^*)$ maximizes $g(x,y)$, completing the proof.
\end{proof}

\subsection{Integrality Gap of LP for MC} \label{sec:integrality-gap}

Using a method of Ageev and Sviridenko~\cite{AS04}, we can lift the symmetry gap of \Cref{thrm:symmetry-gap} to an integrality gap for the LP underlying \Cref{thrm:LP-alg}.

\begin{theorem} \label{thrm:LP-gap}
For every rational constant $0 < c \leq 1$ and $\epsilon > 0$, the linear program in the proof of \Cref{thrm:LP-alg} has an instance with integrality gap at most $(1+\epsilon) \rho(c)$.
\end{theorem}
\begin{proof}
First consider the case $c = \nicefrac1s$, where $s \in \mathbb{N}$. For a parameter $M \in \mathbb{N}$, consider a set system with $\binom{Ms}{s}$ elements $e_I$ of unit weight indexed by subsets $I \subseteq \{1,\dots,Ms\}$ of size $s$. The sets are $S_i = \{ e_I : i \in I \}~\forall i \in \{1,\dots,Ms\}$. The constant solution $y_i = \nicefrac1s~\forall i$ is feasible and has value $\binom{Ms}{s}$ since each $e_I$ is covered exactly $s$ times. Conversely, every integral solution has value $\binom{Ms}{s} - \binom{(1-c)Ms}{s}$: if the solution chooses the sets $S_i$ for $i \in A$, where $|A| = cMs$, then an element $e_I$ is covered unless $I \subseteq \overline{A}$. The integrality gap is thus $1 - \binom{(1-c)Ms}{s}/\binom{Ms}{s}$, which tends (as $M \to \infty$) to $1 - (1-c)^s = \rho(c)$.

Second, suppose that $\nicefrac1{(s+1)} < c < \nicefrac1s$ for $s \in \mathbb{N}$. Let $a,b,p$ be as in the proof of \Cref{thrm:SM-hardness}. For a parameter $M \in \mathbb{N}$, consider a set system with elements of two types: $\binom{Msa}{s}$ elements $e'_I$ of total weight $p$ indexed by subsets of $\{1,\dots,Msa\}$ of size $s$, and $\binom{M(s+1)b}{s}$ elements $e''_J$ of total weight $1-p$ indexed by subsets of $\{1,\dots,M(s+1)b\}$ of size $s+1$; all elements of the same type have the same weight. The sets are $S'_i = \{ e'_I : i \in I \}~\forall i \in \{1,\dots,Msa\}$ and $S''_j = \{ e''_J : j \in J \}~\forall j \in \{1,\dots,M(s+1)b\}$. The solution $y'_i = \nicefrac1s~\forall i$ and $y''_j = \nicefrac1{(s+1)}~\forall j$ is feasible since $Ma + Mb = c(Msa + M(s+1)b)$, and has value $1$.

Consider an integer solution choosing $xMsa$ sets $S'_i$ and $yM(s+1)b$ sets $S''_j$, where $xMsa + yM(s+1)b = M(a+b)$; this is the same constraint as in the definition of $\overline{\mathrm{OPT}}$ in the proof of \Cref{thrm:SM-hardness}. The value of this solution is
\begin{multline*}
 \frac{p}{\binom{Msa}{s}} \left[ \binom{Msa}{s} - \binom{xMsa}{s} \right] + \frac{1-p}{\binom{M(s+1)b}{s+1}} \left[ \binom{M(s+1)b}{s+1) - \binom{xM(s+1)b}{s+1}} \right] = \\ p(1-(1-x)^s) + (1-p)(1-(1-y)^{s+1}) + o_M(1).
\end{multline*}
The proof of \Cref{thrm:SM-hardness} shows that this is at most $\rho(c) + o_M(1)$. This expression bounds the integrality gap.
\end{proof}

\section{Separating MC and SM: Improved Algorithm for MC}\label{sec:separating}
In this section we present an improved algorithm for MC that separates MC and SM.
We start by proving that when $ c=\nicefrac{1}{2}$, MC admits an improved approximation of at least $ 0.7533$, thus separating MC and SM (recall that when $ c=\nicefrac{1}{2}$ SM admits a hardness result of $\rho(\nicefrac12) = \nicefrac{3}{4}$ by \Cref{thrm:SM-hardness}).
Also, we assume that the MC instance is given explicitly. Later on, we comment how to modify the argument to accommodate only value oracle access.

Denote by $ X^*\subseteq [n]$ an optimal solution to MC, and let $ \text{OPT}\triangleq \sum _{e\in \cup _{i\in X^*}S_i}w_e$ be the value of an optimal solution.
We partition the elements $E$ into two collections, depending on how many sets contain an element.
Specifically, let $ m_e$ be the number of sets that contain element $e\in E$, i.e., $ m_e\triangleq |\{ i:e\in S_i\} |$.
Consider the following partition of $E$: $ E_2\triangleq \{ e\in E:m_e=2\}$ and $ E_{\neq 2}\triangleq \{ e\in E:m_e\neq 2\}$.
This partition naturally implies a partition of $ \text{OPT}$ into $ \text{OPT}_2$ and $ \text{OPT}_{\neq 2}$ as follows:
\[
\text{OPT}_2 \triangleq \sum _{e\in E_2\colon e\in \cup _{i\in X^*}S_i}w_e x_e ~\text{ and }~
\text{OPT}_{\neq 2}\triangleq \sum _{e\in E_{\neq 2}\colon e\in \cup _{i\in X^*}S_i}w_e x_e .
\]
Note that $ \text{OPT}=\text{OPT}_2 + \text{OPT}_{\neq 2}$.

Recall from our proof of \Cref{thrm:LP-alg} that our algorithm guarantees for every element $e\in E$ that the probability that $e$ is covered by our algorithm is at least:
\begin{align}
x_e \cdot ( 1-\sigma (\alpha^*,m_e)),\label{minimizerMhalf}
\end{align}
and we chose $\alpha^*=1-\ln{2}$ when $ c=\nicefrac{1}{2}$.
Recall that for this value of $c$ and choice of $\alpha^*$ we had $ m^*(\alpha^*)=2$ which is the $m$ that minimizes \eqref{minimizerMhalf}.
Thus, for every element $e\in E_2$ we have:
\begin{align}
   1-\sigma(\alpha^*,m_e)=1-\sigma(\alpha^*,2)=\rho(\nicefrac{1}{2})=\nicefrac{3}{4}. \label{E2-guarantee}
\end{align}
However, for elements $ e\in E_{\neq 2}$ we can obtain an improved guarantee since for these elements $m_e$ cannot equal the $m$ that minimizes \eqref{minimizerMhalf}.
Specifically, we use the following notation:
\begin{align}
    \rho _{\neq 2}\triangleq \min _{m\geq 1\colon m\neq 2} \{ 1-\sigma(\alpha^*,m)\}, \label{Enot2-guarantee}
\end{align}
\providecommand{\rhohalf}{\rho}
where a straightforward calculation shows that $ \rho _{\neq 2}=1-\sigma(\alpha^*,3)\geq 0.766796$.
We also use the notation $\rhohalf \triangleq \rho(\nicefrac12)$.

We are now ready to prove \Cref{thrm:separating} in the explicit setting.

\begin{proof}[Proof of \Cref{thrm:separating} when the MC instance is given explicitly]
Fix $ \varepsilon >0$ that will be chosen later.
The algorithm can assume that it is in one of two cases: 
\begin{enumerate}[(1)]
\item $\text{OPT}_2 \leq \varepsilon \cdot \text{OPT}$; or \item $ O/(1+\varepsilon)\leq \text{OPT} _2 \leq O$ for some value $O$ known to the algorithm.
\end{enumerate}

To this end, the algorithm first finds a constant $\beta$-approximation $A$ for $\text{OPT}$, e.g., using greedy, satisfying $\beta\cdot\text{OPT}\leq A \leq \text{OPT}$. It then constructs $ O(\ln{(\nicefrac{1}{\varepsilon})/\varepsilon})$ intervals of the form $ [O/(1+\varepsilon),O]$ starting with $[(\nicefrac{A}{\beta})/(1+\varepsilon),\nicefrac{A}{\beta}]$ and ending with $ [(\nicefrac{A}{\beta})/(1+\varepsilon)^j,(\nicefrac{A}{\beta})/(1+\varepsilon)^{j-1}]$, where $j$ is the smallest integer for which $ (\nicefrac{A}{\beta})/(1+\varepsilon)^j\leq \varepsilon \cdot A$. It also constructs the interval $ [0, (\nicefrac{A}{\beta})/(1+\varepsilon)^j]$.
Since $ \text{OPT}_2\leq \text{OPT}$,
necessarily $\text{OPT}_2$ belongs to one of these intervals. If we know that $\text{OPT}_2$ belongs to one of these intervals then one of the two cases above holds. We can ``guess'' which interval $\text{OPT}_2$ belongs to by trying out all possibilities.

%trying all these intervals ensures that for one of them one of the two cases would hold under the assumption that $\text{OPT}_2$ belongs to the interval.

In each of these cases our algorithm behaves differently.
First, let us focus on the case that $ \text{OPT}_2 \leq \varepsilon \cdot \text{OPT}$.
In this case we use our algorithm from the proof of \Cref{thrm:LP-alg} on the instance after we discard all elements in $E_2$. Clearly, the expected value of the output is at least:
\begin{align}
\rho _{\neq 2}\cdot \text{OPT}_{\neq 2} \geq \rho _{\neq 2}\cdot (1-\varepsilon)\text{OPT}, \label{case1}    
\end{align}
since:
\begin{enumerate}[(1)]
\item $ \text{OPT}_{\neq 2}$ is a lower bound on the value of an optimal solution of the instance after we removed all elements in $E_2$; and
\item $\text{OPT}_{\neq 2}=\text{OPT}-\text{OPT}_2\geq (1-\varepsilon )\text{OPT}$.
\end{enumerate}

Second, let us focus on the case that $ O/(1+\varepsilon)\leq \text{OPT} _2 \leq O$ for some value $O$.
In this case we choose the best of the following two solutions:
\begin{enumerate}
    \item Apply \Cref{thrm:MaxkVC} for Max $k$-Vertex-Cover on the instance that contains only the elements of $E_2$.\footnote{The instance only contains elements from $E_2$ and thus it reduces to Max $k$-Vertex-Cover.} \label{alg1}
    \item Apply our algorithm from \Cref{thrm:LP-alg} with an added constraint in the LP, described below.\label{alg2}
\end{enumerate}
When examining \Cref{alg1}, it is easy to see that the expected value of the output is at least:
\begin{align}
    \rho _{\text{VC}}\cdot \text{OPT}_2 ,\label{guarantee1}
\end{align}
where $ \rho_{\text{VC}}\geq 0.94$ as guaranteed by \Cref{thrm:MaxkVC}.

When examining \Cref{alg2}, we add to the LP the constraint: $ \sum _{e\in E_2}w_e x_e \leq O$.
Note that the LP remains feasible and its optimal value is still at least $\text{OPT}$, since we know that: $\text{OPT} _2 \leq O$.
Applying the algorithm from the proof of \Cref{thrm:LP-alg} gives that the expected value of the output is at least:
\begin{align}
 &\rhohalf\sum _{e\in E_2}w_e x_e + \rho _{\neq 2}\sum _{e\in E_{\neq 2}}w_e x_e \nonumber\\ 
 =& \rho _{\neq 2} \sum _{e\in E} w_e x_e - (\rho _{\neq 2}-\rhohalf) \sum _{e\in E_2}w_e x_e \nonumber \\
  \geq & \rho _{\neq 2}\cdot \text{OPT} - (\rho _{\neq 2}-\rhohalf) \cdot O \label{value-ineq1}\\
  \geq & \rho _{\neq 2}\cdot \text{OPT} - (\rho _{\neq 2}-\rhohalf)\cdot (1+\varepsilon)\cdot \text{OPT}_2 \label{value-ineq2}\\
  = & \rho _{\neq 2}\cdot \text{OPT} _{\neq 2} + ((1+\varepsilon ) \rhohalf-\varepsilon \cdot \rho _{\neq 2})\cdot \text{OPT}_2 .\label{guarantee2}
\end{align}
Inequality \eqref{value-ineq1} follows since $ \sum _{e\in E}w_e x_e \geq \text{OPT}$, $\rho _{\neq 2}\geq \rhohalf $, and the additional LP constraint we added.
Moreover, inequality \eqref{value-ineq2} follows from $ O\leq (1+\varepsilon) \cdot \text{OPT}_2$.
Finally, \eqref{guarantee2} follows from the fact that $ \text{OPT}=\text{OPT}_2 + \text{OPT}_{\neq 2}$.

Thus, taking the best of \eqref{guarantee1} and \eqref{guarantee2} results in the following lower bound on the expected value of the output:
\begin{align}
    \max \{ \rho _{\text{VC}}\cdot \text{OPT}_2~ ,~ \rho _{\neq 2}\cdot \text{OPT} _{\neq 2} + (\rhohalf - \varepsilon(\rho_{\neq2}-\rhohalf))\cdot \text{OPT}_2\} .\label{guarantee3}
\end{align}
Guarantee \eqref{guarantee3} can be further lower bounded by taking a convex combination of the two guarantees it encompasses. Choosing the convex combination which makes the coefficients of $\text{OPT}_2$ and $ \text{OPT}_{\neq 2}$ equal, and recalling that $ \text{OPT} = \text{OPT}_2+\text{OPT}_{\neq 2}$, results in the following lower bound on the expected value of the algorithm:
\[ \frac{\rho _{VC}\cdot \rho _{\neq 2}}{\rho _{\text{VC}}+(1+\varepsilon)\cdot (\rho _{\neq 2}-\rhohalf)} \cdot \text{OPT}.\]
Plugging the values of $\rhohalf$, $ \rho _{\neq 2}$, and $ \rho _{\text{VC}}$, as well as choosing, e.g., $\varepsilon = \nicefrac{1}{n}$, yields an approximation of $ 0.7533 - o(1)$, as required.
\end{proof}

We can extend \Cref{thrm:separating} to the value oracle setting along the lines of \Cref{thrm:LP-alg-oracle}.

\begin{proof}[Proof of \Cref{thrm:separating} when the MC instance is given by a value oracle]

As in the proof of \Cref{thrm:LP-alg-oracle}, we can determine in time $O(n^2)$ the total weight of all elements covered exactly by $S_i,S_j$ for all $i,j$. This allows us to explicitly construct the instance containing only the elements of $E_2$, as well as provide value oracle access to the instance containing the remaining elements. In the rest of the proof, instead of applying \Cref{thrm:LP-alg} on the latter instance, we apply \Cref{thrm:LP-alg-oracle}.
\end{proof}

In \Cref{app:SDP} we show, assuming a conjecture related to the SDP used by \cite{ABG16,BHZ24}, that MC admits an improved approximation not only for $ c=\nicefrac{1}{2}$ but (empirically) also for every $\nicefrac{1}{2}\leq c <1 $: plotting the performance $\rho_{\text{SDP}}(c)$ of the improved algorithm against $\rho(c)$ suggests that $\rho_{\text{SDP}}(c) > \rho(c)$ for all $\nicefrac{1}{2}\leq c < 1$.
If this conjecture is indeed true, then this separates MC and SM not only when $ c=\nicefrac{1}{2}$ but for every $ \nicefrac{1}{2}\leq c <1$.

We are unable to extend this to lower values of $c$ (even conditionally) since the SDP only handles elements covered at most twice, and these elements are the hardest for the LP, i.e., the approximation achieved for these elements is the lowest.

\subsection{Improved Approximation for MC when \texorpdfstring{$ \nicefrac{1}{2} < c <1$}{1/2 < c < 1}}\label{app:SDP}

In order to handle the case when $ \nicefrac{1}{2} < c <1$, we formulate an SDP relaxation for the special case of MC where each element $ e\in E$ belongs to at most two sets.
Recalling that $m_e$ denotes the number of sets containing $e$, if $ m_e=1$ we denote the (single) set that contains $e$ by $ S_{i_e}$ and if $ m_e=2$ we denote the two sets that contain $e$ by $ S_{i_{e,1}}$ and $ S_{i_{e,2}}$.
Moreover, let $E_1=\{ e\in E:m_e=1\}$ and $ E_2=\{ e\in E:m_e=2\}$.

Following the line of work \cite{ABG16,LLZ02,RT12}, we consider the following SDP formulation:
\begin{align}
\max ~~~ &\sum _{e\in E_1}w_e \| \mathbf{u}_{i_e}\|^2_2 + \sum _{e\in E_2}w_e (\| \mathbf{u}_{i_{e,1}}\|_2^2 + \| \mathbf{u}_{i_{e,2}}\|_2^2 - \mathbf{u}_{i_{e,1}}\cdot \mathbf{u}_{i_{e,2}} ) & \nonumber \\
{\text{s.t.}}~~~ & \| \mathbf{u}_0\|_2^2 = 1 & \nonumber\\
& \mathbf{u}_0 \cdot \mathbf{u}_i = \| \mathbf{u}_i\|_2^2 & \forall i\in [n] \nonumber\\
& \mathbf{u}_i \cdot \mathbf{u}_j \leq \min \{ \| \mathbf{u}_i\|_2^2, \| \mathbf{u}_j\| _2^2\} & \forall i,j\in [n] \nonumber\\
& \mathbf{u}_i\cdot \mathbf{u}_j \geq \max \{ 0,\| \mathbf{u}_i\|_2^2 + \| \mathbf{u}_j\| _2^2 - 1\} & \forall i,j\in [n] \nonumber\\
& \sum _{i=1}^n \| \mathbf{u}_i\|_2^2 \leq cn &\nonumber 
\end{align}
This relaxation contains a special unit vector denoted by $ \mathbf{u}_0$, and a vector $\mathbf{u}_i$ for every set $S_i$, $ i\in [n]$.

The following lemma states that this SDP formulation is indeed a relaxation of MC.
We denote by $\text{OPT}$ the value of an optimal solution and $\text{OPT}_{\text{SDP}}$ the value of an optimal solution to the SDP relaxation
\begin{lemma}\label{lem:SDP-relaxation}
For every instance of MC satisfying $\max _{e\in E}\{ m_e\}\leq 2$, $\text{OPT}_{\text{SDP}} \geq \text{OPT} $.
\end{lemma}
\begin{proof}
Let $X^*\subseteq [n]$ be an optimal solution.
Choose $ \mathbf{u}_0=1$, and for every $ i\in X^*$, set $ \mathbf{u}_i=1$ and $ \mathbf{u}_i=0$ otherwise.
Clearly, all constraints are satisfied.
Focusing on the objective function:
\begin{enumerate}[(1)]
\item If $e\in E_1$ then $e$'s contribution to the objective equals $1$ if and only if $i_e\in X^*$ (and $0$ otherwise). \item If $e\in E_2$ then $e$'s contribution to the objective equals $1$ if and only if $ i_{e,1}\in X^*$ or $ i_{e,2}\in X^*$ (and $0$ otherwise).
\end{enumerate}
Hence, the objective value of the above solution equals $ \text{OPT}$.
This concludes the proof.
\end{proof}

The presentation of the above SDP formulation can be simplified using additional notations.
First, we denote $\mu _i \triangleq \| \mathbf{u}_i\|_2^2$, and note that $ \mu _i$ can be intuitively seen as the marginal probability of the set $S_i$ being chosen by the SDP solution.
Second, we denote $ \rho _{i,j}\triangleq \mathbf{u}_i \cdot \mathbf{u}_j$, and note that $\rho _{i,j}$ can be intuitively seen as the joint probability of both sets $S_i$ and $S_j$ being chosen by the SDP solution.
Using these notations, the SDP can be rewritten as follows:
\begin{align}
\max ~~~ &\sum _{e\in E_1}w_e \mu _{i_e} + \sum _{e\in E_2}w_e (\mu _{i_{e,1}}+\mu _{i_{e,2}}-\rho _{i_{e,1},i_{e,2}} ) & \nonumber \\
{\text{s.t.}}~~~ & \| \mathbf{u}_0\|_2^2 = 1 & \nonumber\\
& \mathbf{u}_0 \cdot \mathbf{u}_i = \mu _i & \forall i\in [n] \nonumber\\
& \rho _{i,j}\leq \min \{ \mu _i, \mu _j\} & \forall i,j\in [n] \nonumber\\
& \rho _{i,j} \geq \max \{ 0,\mu _i + \mu _j - 1\} & \forall i,j\in [n] \nonumber\\
& \sum _{i=1}^n \mu _i \leq cn &\nonumber 
\end{align}

Our improved algorithm, similarly to Raghavendra and Tan \cite{RT12}, solves the $\text{poly}(\nicefrac{1}{\varepsilon})$-round Lasserre hierarchy of the above SDP, for some small constant $\varepsilon >0$.
Following \cite{RT12}, Austrin, Benabbas, and Georgiou \cite[Lemma 3.1]{ABG16} apply a conditioning procedure that produces a solution $ \{ \mathbf{u}_i\} _{i=0}^n$ to the SDP formulation that satisfies the following conditions:
\begin{enumerate}
    \item $ \{ \mathbf{u}_i\} _{i=0}^n$ is feasible. \label{SDP-feasible}
    \item The objective value of $ \{ \mathbf{u}_i\} _{i=0}^n$ is at least $ (1-\varepsilon) \cdot \text{OPT}$.\label{SDP-value}
    \item $\epsilon \leq \mu_1,\dots,\mu_n \leq 1-\epsilon$.
    \item $ \{ \mathbf{u}_i\} _{i=0}^n$ satisfies the {\em global decorrelation} property:
    \[ \frac{1}{n^2}\sum _{i\neq j} |\tilde{\rho } _{i,j}|\leq \varepsilon ,\]
    where $\tilde{\rho}_{i,j}\triangleq (\rho _{i,j}-\mu _i \mu _j)/(\sqrt{\mu _i (1-\mu _i)}\sqrt{\mu _j (1-\mu _j)})$.\label{SDP-global}
\end{enumerate}

We note that $ \tilde{\rho}$ has a simple and elegant geometric interpretation.
Let us write $\mathbf{u}_i$ as a component in the direction of $ \mathbf{u}_0$ and a component in the direction orthogonal to $ \mathbf{u}_0$.
Specifically, $\mathbf{u}_i=\mu _i \mathbf{u}_0 + \mathbf{u}_i^{\perp}$, where $\mathbf{u}_i^{\perp}$ is a vector that is orthogonal to $ \mathbf{u}_0$.
Denote by $ \tilde{\mathbf{u}}_i^{\perp}$ the normalized vector of $ \mathbf{u}_i^{\perp}$, i.e., $ \tilde{\mathbf{u}}_i^{\perp}\triangleq  \mathbf{u}_i^{\perp}/\|\mathbf{u}_i^{\perp}\|_2$.
A simple calculation shows that $ \| \mathbf{u}_i^{\perp}\|_2^2=\mu _i (1-\mu _i)$, and thus: $\tilde{\mathbf{u}}_i^{\perp}\triangleq  \mathbf{u}_i^{\perp}/\sqrt{\mu _i (1-\mu _i)} $.
Therefore, we can conclude that $\tilde{\rho}_{i,j}=\tilde{\mathbf{u}}_i^{\perp} \cdot \tilde{\mathbf{u}}_j ^{\perp}$.

Our algorithm follows the same outline as \cite{LLZ02,RT12}, but similarly to the proof of \Cref{thrm:LP-alg}, it changes the marginal probabilities (as given by the SDP).
Formally, the altered marginal probability, which averages the probability given by the SDP's solution and a random uniform solution ($c$), is defined as follows:
\[ \eta _i \triangleq \alpha c + (1-\alpha) \mu _i.\]
Here, as before, $0\leq \alpha\leq 1$ is a parameter that will be chosen later.
The algorithm outputs a solution $X\subseteq [n]$ by picking a random gaussian in the space orthogonal to $ \mathbf{u}_0$, projecting the normalized component of each $ \mathbf{u}_i$ in this space on the gaussian, and choosing $S_i$ if it falls in the lower gaussian tail that has a probability mass of $\eta _i$.
Formally:
\[ X\leftarrow \{ i\in [n]: \tilde{\mathbf{u}}_i^{\perp}\cdot \mathbf{g} \leq \Phi ^{-1}(\eta _i)\} .\]
Here, $ \mathbf{g}$ is a random gaussian vector in the space orthogonal to $\mathbf{u}_0$, and $\Phi$ is the standard gaussian cumulative distribution function.
It is important to note that $X$ is not necessarily a feasible solution.
However, if this happens our algorithm chooses uniformly at random $cn$ out of the sets in $X$.

\paragraph{Analyzing the Value of the Output $X$.}
First, consider an elements $ e\in E_1$.
Our goal is to lower bound the ratio of the probability that $e$ is covered in the solution, i.e., $ i_e\in X$, and the contribution of $e$ to the objective of the SDP.
The former equals $ \eta _{i_e}=\alpha c + (1-\alpha)\mu _{i_e}$, and latter equals $\mu _{i_e}$.
Thus, if $ r(\alpha,1)$ is the worst approximation achieved for an elements in $E_1$, we get that:
\begin{align}
   r(\alpha,1)\triangleq \min _{0\leq \mu \leq 1}\left\{ \frac{\alpha c + (1-\alpha) \mu}{\mu}\right\}=\alpha c +1 - \alpha. \label{SDP-Alg-1}
\end{align}

Second, consider an elements $ e\in E_2$.
As before, our goal is to lower bound the ratio of the probability that $e$ is covered in the solution, i.e., $ i_{e,1}\in X$ or $ i_{e,2}\in X$, and the contribution of $e$ to the objective of the SDP.
The latter equals $ \mu _{i_{e,1}}+\mu _{i_{e,2}}-\rho _{i_{e,1},i_{e,2}}$.

Let us focus on the former:
\begin{align}
& \Pr [i_{e,1}\in X ~\text{ or }~ i_{e,2}\in X]  \nonumber \\
 = & 1-\Pr [\tilde{\mathbf{u}}_{i_{e,1}}^{\perp}\cdot \mathbf{g}> \Phi ^{-1}(\eta _{i_{e,1}}),\mathbf{u}_{i_{e,2}}^{\perp}\cdot \mathbf{g}> \Phi ^{-1}(\eta _{i_{e,2}}) ]   \label{SDP-Alg-Obj1}\\
= &1-\Pr [\tilde{\mathbf{u}}_{i_{e,1}}^{\perp}\cdot \mathbf{g}< \Phi ^{-1}(1-\eta _{i_{e,1}}),\mathbf{u}_{i_{e,2}}^{\perp}\cdot \mathbf{g}< \Phi ^{-1}(1-\eta _{i_{e,2}}) ] .\label{SDP-Alg-Obj2}
\end{align}
In the above, equality \eqref{SDP-Alg-Obj1} follows from the definition of the algorithm, and equality \eqref{SDP-Alg-Obj2} follows from the facts that $ -\Phi^{-1}(t)=\Phi^{-1}(1-t)$ and that $-\mathbf{g}$ and $ \mathbf{g}$ have the same distribution.
Note that $ (\tilde{\mathbf{u}}_{i_{e,1}}^{\perp}\cdot \mathbf{g},\mathbf{u}_{i_{e,2}}^{\perp}\cdot \mathbf{g})$ is a bivariate gaussian random variable distributed with the following parameters:
\[ \left( \begin{matrix} \tilde{\mathbf{u}}_{i_{e,1}}^{\perp}\cdot \mathbf{g} \\ \mathbf{u}_{i_{e,2}}^{\perp}\cdot \mathbf{g} \end{matrix}\right) \sim N \left(\left(\begin{matrix} 0 \\ 0\end{matrix}\right), \left( \begin{matrix} 1 & \tilde{\rho}_{i_{e,1},i_{e,2}} \\  \tilde{\rho}_{i_{e,1},i_{e,2}} & 1\end{matrix}\right) \right) .\]
Using the following notation: $ \Phi _2 (\rho,\eta _1, \eta _2)\triangleq \Pr [X< \Phi ^{-1}(\eta _1),Y < \Phi ^{-1}(\eta_2)]$, where:
\[ \left( \begin{matrix} X \\ Y \end{matrix}\right) \sim N \left(\left(\begin{matrix} 0 \\ 0\end{matrix}\right), \left( \begin{matrix} 1 & \rho \\  \rho & 1\end{matrix}\right) \right),\]
we can rewrite \eqref{SDP-Alg-Obj2} as follows: $ 1-\Phi _2 (\tilde{\rho}_{i_{e,1},i_{e,2}},1-\eta _{i_{e,1}},1-\eta _{i_{e,2}})$.
Hence, if $r(\alpha,2)$ is the worst approximation achieved for an elements in $E_2$, we get that:
\begin{align}
    r(\alpha,2)\triangleq \min _{(\tilde{\rho},\rho,\mu_1,\mu_2)\in \mathcal{F}} \left\{ \frac{1-\Phi_2(\tilde{\rho},1-\alpha c - (1-\alpha)\mu _1,1-\alpha c - (1-\alpha )\mu _2)}{\mu _1 + \mu_2-\rho}\right\}. \label{SDP-Alg-2}
\end{align}
In the above, $\mathcal{F}$ encodes the collection of feasible configurations of $ (\tilde{\rho},\rho,\mu_1,\mu_2)$ as determined by the SDP.
Thus:
\begin{align*}
\mathcal{F}=\{ (\tilde{\rho},\rho,\mu_1,\mu_2): \, & 0\leq \mu_1, \mu_2 \leq 1,\\
&\rho\leq \min \{ \mu_1,\mu_2\}, \rho\geq \max \{ 0,\mu_1+\mu_2-1\},\\ &\tilde{\rho}=(\rho - \mu_1\mu_2)/(\sqrt{\mu_1(1-\mu_1)}\sqrt{\mu_2(1-\mu_2)}) \}.
\end{align*}

Therefore, the overall lower bound on the approximation factor is achieved by taking the worst of $ r(\alpha,1)$ \eqref{SDP-Alg-1} and $ r(\alpha,2)$ \eqref{SDP-Alg-2}: $r(\alpha)\triangleq \min \{ r(\alpha,1),r(\alpha,2)\}$.
Combining this with \eqref{SDP-value}, we get that:
\[ \mathbb{E} \left[ \sum _{e\in \cup _{i\in X}S_i} w_e\right]\geq (1-\varepsilon) \cdot r(\alpha) \cdot \text{OPT}.\]

\vspace{3pt}
\noindent {\bf Analyzing the Size of $X$.}
Denote by $X_i$ the indicator that $ i\in X$, i.e., $S_i$ is chosen to the solution, and let $ X\triangleq \sum _{i=1}^n X_i$ be the number of sets chosen.
Clearly, $ \mathbb{E}[X]\leq cn$.
Our goal is to prove that $X$ is concentrated around its mean.

First, let us bound the variance of $X$:
\begin{align}
\text{Var}(X)=\sum  _{i=1}^n \text{Var}(X_i)+\sum _{i\neq j}\text{Cov}(X_i,X_j).\label{var1}
\end{align}
The first term on the right hand side of \eqref{var1} can be easily upper bounded by $cn$, since $ \text{Var}(X_i)\leq \mathbb{E}[X_i]$, $ \forall i\in [n]$, and $\mathbb{E}[X]\leq cn$.
The second term on the right hand side of \eqref{var1} can be bounded as follows:
\begin{align}
\sum _{i\neq j}\text{Cov}(X_i,X_j)= \sum _{i\neq j}\left(\Phi _2(\tilde{\rho}_{i,j},\eta _i,\eta _j) -\eta_i \eta_j\right) \leq 2\sum _{i\neq j} |\tilde{\rho }_{i,j}| \leq 2\varepsilon n^2.
\end{align}
The penultimate inequality follows from, e.g., \cite[Lemma 2.7]{ABG16}, who showed that: $\Phi _2(\rho,\eta_1,\eta_2)\leq \eta_1 \eta_2 + 2 |\rho|$.
Moreover, the last inequality follows from the global decorrelation property \eqref{SDP-global}.
Hence, $\text{Var}(X)\leq cn+2\varepsilon n^2 = O(\varepsilon n^2)$.\footnote{If $n < \nicefrac1\varepsilon$ then we solve the instance optimally by brute force, and otherwise $\operatorname{Var}(X) \leq 3\varepsilon n^2$.}
Applying Chebyshev's inequality yields for every $\delta >0$:
\begin{align}
    \Pr [X\geq (1+\delta)cn]\leq \text{Var}(X)/(\delta cn)^2=O(\varepsilon \delta ^2).
\end{align}
Thus, choosing, e.g., $\delta  = \varepsilon ^{\nicefrac{1}{3}}$, gives that with a probability of at least $1-O(\varepsilon ^{\nicefrac{1}{3}})$ we have: $X\leq (1+\varepsilon ^{\nicefrac{1}{3}})cn$.

\vspace{3pt}
\noindent {\bf Combining Value and Size of X Simultaneously.}
We execute the algorithm until we obtain a solution $X$ of size smaller than $ (1+\varepsilon ^{\nicefrac{1}{3}})cn$.
Therefore, from the law of total expectation:
\begin{align}
&\mathbb{E}\left[\sum _{e\in \cup _{i\in X}}w_e \;\middle|\; X<(1+\varepsilon ^{\nicefrac{1}{3}})cn\right]  \nonumber\\
\geq &\mathbb{E}\left[\sum _{e\in \cup _{i\in X}}w_e\right]-\mathbb{E}\left[\sum _{e\in \cup _{i\in X}}w_e \;\middle|\; X\geq (1+\varepsilon ^{\nicefrac{1}{3}})cn\right]\cdot \Pr \left[X\geq (1+\varepsilon ^{\nicefrac{1}{3}}cn)\right]  \nonumber\\
\geq & (1-\varepsilon)\cdot r(\alpha) \cdot \text{OPT}-\frac{\text{OPT}}{c}\cdot O(\varepsilon ^{\nicefrac{1}{3}})
\label{SDP-combined}\\
= & (1-O(\varepsilon ^{\nicefrac{1}{3}}))\cdot r(\alpha)\cdot \text{OPT}.\nonumber
\end{align}
Inequality \eqref{SDP-combined} follows from the upper bound on the probability of the bad event that $X$ is at least $(1+\varepsilon ^{\nicefrac{1}{3}})cn$, and from the fact that $\text{OPT}\geq c|E|$.

\vspace{3pt}
\noindent {\bf Completing the Analysis.}
So far our algorithm obtains a solution $X$ of size smaller than $ (1+\varepsilon ^{\nicefrac{1}{3}})cn$ whose value is at least $ (1-O(\varepsilon ^{\nicefrac{1}{3}}))\cdot r(\alpha) \cdot \text{OPT}$.
If $ |X| > cn$, then we uniformly at random choose a subset of size $ cn$ of $X$.
This incurs an additional multiplicative loss of $ (1-O(\varepsilon^{\nicefrac{1}{3}}))$ in the value of the solution since each element remains covered with a probability of at least $ cn/|X|\geq 1-O(\varepsilon^{\nicefrac{1}{3}})$.
Hence, the value of the output is at least $ (1-O(\varepsilon^{\nicefrac{1}{3}}))^2 \cdot r(\alpha)\cdot \text{OPT}\geq (1-O(\varepsilon^{\nicefrac{1}{3}})) \cdot r(\alpha)\cdot \text{OPT}$.

Choosing $\varepsilon >0$ sufficiently small allows us to obtain an approximation guarantee that is as close to $ r(\alpha)$ as desired.
Denote $\rho _{\text{SDP}}(c)\triangleq \max _{0\leq \alpha \leq 1}\{ r(\alpha)\}$.

\paragraph{Estimating $\rho _{\text{SDP}}(c)$.}
Lower bounding $ r(\alpha,2)$ \eqref{SDP-Alg-2} is not straightforward.
To this end we state a generalization of the \emph{simplicity conjecture} proposed by Austrin \cite{A07}.
\begin{conjecture}
\label{conj:simplicity}
For every $ 0<c<1$ and $ 0\leq \alpha \leq 1$, the minimizer of $ r(\alpha,2)$ is achieved for a configuration $ (\tilde{\rho},\rho,\mu_1,\mu_2) \in \mathcal{F}$ satisfying: $ \mu_1=\mu_2=\mu$ and $ \rho = \max \{ 0,2\mu -1\}$.
\end{conjecture}
It is worth mentioning that \Cref{conj:simplicity} was proved correct by Brakensiek, Huang, and Zwick \cite{BHZ24} for $ c=\nicefrac{1}{2}$ and the particular value of $ \alpha$ that is optimal for this value of $c$.

Assuming \Cref{conj:simplicity}, we can numerically compute $ \rho _{\text{SDP}}(c)$ for every $ \nicefrac{1}{2}<c<1$.
Using this value one can repeat a similar analysis to the one appearing in \Cref{sec:separating}.
This yields the approximation one can obtain for MC when $ \nicefrac{1}{2}<c<1$ (as seen in \Cref{fig:curves}).

\section{Open Questions}\label{sec:open-questions}

\paragraph{What is the optimal approximation ratio for SM?} \Cref{thrm:contgreedy-alg} gives an approximation ratio of $1-(1-c)^{\nicefrac1c}$, which is matched by \Cref{thrm:SM-hardness} when $\nicefrac1c \in \mathbb{N}$. We conjecture that $1-(1-c)^{\nicefrac1c}$ is the optimal approximation ratio for all $c$.

\paragraph{What is the optimal approximation ratio for MC?} We conjecture that the SDP-based algorithm for Max $k$-Vertex-Cover, used to prove \Cref{thrm:separating}, can be extended to give a $1-\nicefrac1e$ approximation for MC for arbitrary $k$, as well as to yield the optimal approximation ratio for $k=cn$ for all $c$.

\bibliographystyle{alpha}
\bibliography{mybibliography}

\end{document}